\documentclass{llncs}
\usepackage[theorems=false,natbib=false]{ifiseriesreplacement}
\usepackage{subfig}
\usepackage{sidecap}
\usepackage{localKit}

\usepackage{amsthm}
\usepackage{thmtools}
\declaretheorem[name=Proposition]{thmprop}
\declaretheorem[name=Lemma,sibling=thmprop]{thmlem}
\declaretheorem[name=Remark,sibling=thmprop]{thmrem}
\declaretheorem[name=Corollary,sibling=thmprop]{thmcor}
\declaretheorem[name=Theorem]{thmthm}
\hypersetup{%
  final=true,%
  colorlinks=true,%
  linkcolor=DarkSlateGrey,%
  citecolor=DarkSlateGrey,%
  urlcolor=DarkSlateGrey,%
}

\newcommand{\sep}{\mathrm{sep}}
\renewcommand{\rel}{\mathrm{rel}}
\newcommand{\Rel}{R}
\newcommand{\nE}{\complement E}
\newcommand{\lev}{\operatorname{lev}}
\newcommand{\Bi}{\mathsf{B}}
\newcommand{\Uni}{\mathsf{U}}
\newcommand{\ones}[1]{\lvert{#1}\rvert}
\newcommand{\nn}{\nu}
\renewcommand{\SC}{C}

\newcommand{\Emax}{E_{\max}}
\newcommand{\tV}{\widetilde{V}}
\newcommand{\tE}{\widetilde{E}}
\newcommand{\tG}{\widetilde{G}}
\renewcommand{\OPT}{\mathrm{OPT}}

\pgfdeclarelayer{background}
\pgfsetlayers{background,main}
\tikzset{%
vertex/.style={circle,fill=black!20,inner sep=0pt,minimum size=18pt},%
small vertex/.style={circle,fill=black!20,inner sep=0pt,minimum size=9pt},%
player/.style={circle,fill=black!20,inner sep=0pt,minimum size=18pt},%
marked player/.style={circle,draw=black,fill=white,inner sep=0pt,minimum size=17pt},%
dirlink/.style={draw,postaction={decorate,decoration=%
	{markings,mark=at position -.8pt with {\arrow[line width=1pt]{stealth'}}}}},%
subgraph/.style={draw,circle,decorate,decoration={coil,aspect=0},minimum size=2cm},%
every picture/.style={>=stealth',line width=.8pt},%
highlight/.style={draw,line width=5pt,-,black!30},%
on grid,%
auto,%
}

\begin{document}
\setlength{\abovedisplayshortskip}{-.5\baselineskip plus .5\baselineskip}
\setlength{\abovedisplayskip}{\medskipamount}
\setlength{\belowdisplayshortskip}{\medskipamount}
\setlength{\belowdisplayskip}{\medskipamount}

\title{The Price of Anarchy in Bilateral Network Formation in an Adversary Model}
\author{Lasse~Kliemann}
\institute{Christian-Albrechts-Universität zu Kiel\\
Institut für Informatik\\
Christian-Albrechts-Platz 4\\
24118 Kiel, Germany\\
\texttt{lki@informatik.uni-kiel.de}}
\maketitle

\begin{abstract}
  We study network formation with the bilateral link formation rule (Jackson and Wolinsky~1996)
  with $n$ \term{players} and \term{link cost} $\alpha>0$.
  After the network is built,
  an adversary randomly destroys one link according to a certain probability distribution.
  Cost for player $v$ incorporates 
  the expected number of players to which $v$ 
  will become disconnected.
  This model was previously studied for unilateral link formation (K.~2011).
  \par
  We prove existence of \term{pairwise Nash equilibria}
  under moderate assumptions on the adversary and $n\geq 9$.
  As the main result, we prove bounds on the \term{price of anarchy}
  for two special adversaries:
  one destroys a link chosen uniformly at random,
  while the other destroys a link 
  that causes a maximum number of player pairs to be separated.
  We prove bounds tight up to constants, 
	namely $O(1)$ for one adversary (if $\al>\frac{1}{2}$),
  and $\Theta(n)$ for the other
  (if $\al>2$ considered constant and $n \geq 9$).
  The latter is the worst that can happen
  for any adversary in this model (if $\al=\Om(1)$).
\end{abstract}

\noindent\textbf{Bibliographic Note.}
This extended abstract has been accepted 
at the \textit{6th International Symposium on Algorithmic Game Theory} (SAGT 2013).

\section{Network Formation Games}
\label{sec:framework}
Network formation games are strategic games with a certain structure
based on players representing vertices in a graph;
we will sometimes call a player also a \enquote{vertex}.
The set of players is $V=\setn{n}\df\setft{1}{n}$, $n \geq 3$,
and the strategy space for each player is $\OI^n$.
A strategy profile $S=\fami{S_v}{v\in V} \in \OI^{n\times n}$ 
determines an undirected graph $G(S)=(V,E(S))$
by the application of a link formation rule.
This rule is a parameter of the game.
Two well-known rules are \term{unilateral link formation} (ULF)
and \term{bilateral link formation} (BLF),
the latter being the main topic of this work.
Under ULF, the built graph is $G(S)=G^\Uni(S)=(V,E^\Uni(S))$, where
\begin{equation*}
  E^\Uni(S)\df\setst{\set{v,w} \in \textstyle{{V \choose 2}}}{S_{vw}=1\lor S_{wv}=1}\period
\end{equation*}
Under BLF, the built graph is $G(S)=G^\Bi(S)=(V,E^\Bi(S))$, where
\begin{equation*}
  E^\Bi(S)\df\setst{\set{v,w} \in \textstyle{{V \choose 2}}}{S_{vw}=1\land S_{wv}=1}\period
\end{equation*}
We omit the \enquote{$\Uni$} and \enquote{$\Bi$} superscripts from our notation
if the link formation rule is clear from context or if we refer to no specific one.
\par
If $S_{vw}=1$ then this is interpreted as the request by player $v$ for a link to player $w$.
Under ULF, this request is enough to have the link built.
Under BLF, also the other player must request it, otherwise it is not built.
Entries $S_{vv}$ are of no concern.
Denote by $\nE(S) \df {V \choose 2} \setminus E(S)$ the links that are \emphasis{not} there.
To determine players' costs, we need two more parameters:
\term{link cost} $\al > 0$ and an \term{indirect cost function} $I_v$ for each player $v$,
defined on all undirected graphs on $V$.
Player $v$'s cost under strategy profile $S$ is then
\begin{equation*}
  C_v(S) \df \ones{S_v} \, \al + I_v(G(S)) \comma
\end{equation*}
where $\ones{S_v}$ denotes the number of $1$s in $S_v$, \ie $\ones{S_v} = \sum_{w\in V} S_{vw}$.
The first term, $\ones{S_v} \, \al$, is called $v$'s \term{building cost},
and $I_v(G(S))$ is called her \term{indirect cost};
we also write $I_v(S) \df I_v(G(S))$.
An example for $I_v$ is $I_v(G) = \sum_{w\in V} \dist_G(v,w)$~\cite{FLM+03}, 
the \term{sum-distance model}, but many others are conceivable.
Oftentimes, $I_v(G) = \infty$ if $G$ is disconnected.
We use \term{total cost} as the social cost, 
namely our \term{social cost} is $\SC(S) \df \sum_{v\in V} C_v(S)$.
For fixed parameters $n$,~$\al$, link formation rule, and indirect cost function,
we call $S$ \term{optimal} if it has minimum social cost among all strategy profiles.
Denote $\OPT$ that minimum social cost.
\par
Strategy profile $S$ is called a \term{Nash equilibrium} (NE) if
\begin{equation*}
  C_v(S) \leq C_v(S_{-v},X) \quad\quad \forall v \in V \quad \forall X \in \OI^n \comma
\end{equation*}
where as usual $(S_{-v},X)$ denotes the strategy profile resulting from $S$
by replacing $v$'s strategy by $X$.
So a NE is characterized by no player having an incentive to deviate from her current strategy,
assuming the strategies of the other players fixed.
For $v,w \in V$ denote $S+vw$ the strategy profile $S'$ with $S'_{vw} = 1$ 
and otherwise identical to $S$.
For $v,w \in V$ denote $S-vw$ the strategy profile $S'$ with $S'_{vw} = 0$ 
and otherwise identical to $S$.
Strategy profile $S$ is called a \term{pairwise Nash equilibrium} (PNE) if it is a NE 
and additionally
\begin{equation}
  \label{eqn:pne}%
  \small%
  C_v(S+vw+wv) \leq C_v(S) \implies C_w(S+vw+wv) > C_w(S)
  \quad\forall \set{v,w} \in \nE(S)
  \period
\end{equation}
So each missing link requires the additional justification that it would be an impairment
for at least one of its endpoints.
Strategy profile $S$ is called \term{pairwise stable} (PS) if \eqref{eqn:pne} holds and
\begin{equation}
  \label{eqn:ps}%
  \small%
  C_v(S-vw) \geq C_v(S)
  \quad\forall \set{v,w} \in E(S)
  \period
\end{equation}
So pairwise stability (PS) is only concerned with single-link deviations.
We call a strategy profile $S$ \term{essential} under ULF if $S_{vw}=1$ implies $S_{wv}=0$;
we call it \term{essential} under BLF if $S_{vw}=1$ implies $S_{wv}=1$.
Since a player has to pay $\al$ for each \emphasis{requested} link,
optima and equilibria of any of the three kinds (NE, PNE, and PS) are essential
and we will thus \emphasis{limit our studies to essential strategy profiles in the following}.
With BLF, an essential strategy profile $S$ is completely determined by the built graph $G(S)$,
and we call a graph $G$ PS if $G=G(S)$ for some PS strategy profile $S$.
\par
NE is well suited for ULF.
However, it is less well suited for BLF since the strategy profile $S=\mathbf{0}$,
\ie where no player issues any requests, is a NE under BLF;
indeed, due to the link formation rule, given $S=\mathbf{0}$ 
no player can make a change to her strategy that would have an effect on the graph built.
Even more, under BLF no player can unilaterally,
\ie by changing her strategy while strategies of all other players are maintained,
build a link.
Under ULF, links can be built unilaterally.
Removal of links can happen unilaterally with ULF and BLF.
We say that a player \term{removes} or \term{sells} a link present in the built graph
if she changes her strategy so that this link will be no longer part of the built graph.
Removing or selling a link will make her building cost smaller by the amount of~$\al$.
We say that players \term{build}, \term{add}, or \term{buy} a link not present in the built graph
if they change their strategies so that this link will be part of the built graph.
With BLF, this can only happen when the two endpoints of the link agree on it,
and then building cost for each of them increases by $\al$.
\par
Since this work is concerned with BLF, we will use PNE and PS.
Clearly, if $S$ is a PNE then $S$ is PS.
The converse holds if cost is convex
on the set of PS strategy profiles~\cite{CP05}.
Convexity of cost relates removal of multiple links
to removal of each of those links alone.
This addresses the difference between PNE and PS:
in the former, removal of multiple links has to be considered,
whereas the latter is only concerned with removal of single links.
Let $v\in V$ and $S$ a strategy profile.
We call $C_v$ \term{convex in $S$}
if for all $k\in\setn{n}$ and $\set{w_1,\hdots,w_k}\subseteq V$ we have
$C_v\parens{S-vw_1-\hdots-vw_k} - C_v(S) \geq \sum_{i=1}^k \parens{C_v(S-(v,w_i)) - C_v(S)}$,
or, equivalently,
\begin{equation}
  \label{eqn:def-convex-2}%
  I_v\Parens{S-vw_1-\hdots-vw_k} - I_v(S)
  \geq \sum_{i=1}^k \Parens{I_v(S-vw_i) - I_v(S)}\period
\end{equation}
We say that \term{cost is convex} if $C_v$ is convex for each $v$ 
and in each strategy profile~$S$.
\par
For fixed parameters $n$,~$\al$, link formation rule, indirect cost function,
and equilibrium concept (NE, PNE, PS), 
the \term{price of anarchy}~\cite{Pap01} is defined as
\begin{equation*}
  \max_{\text{$S$ is an equilibrium}} \, \frac{\SC(S)}{\OPT}\period
\end{equation*}

\section{Adversary Model}
An \term{adversary} is a mapping
assigning to each graph $G=(V,E)$
a probability measure $\P_G$ on the links $E$ of $G$.
Given a connected graph $G$,
the \term{relevance} of a link $e$ for a player $v$ is the number 
of vertices that can, starting at $v$, \emphasis{only} be reached via~$e$.
We denote the relevance of $e$ for $v$ by $\rel_G(e,v)$
and the sum of all relevances for a player 
by $\Rel_{G}(v)\df\sum_{e\in E} \rel_{G}(e,v)$.
If $P=(v,\hdots,w)$ is a path, then
denote $\rel(P,v) \df \sum_{e\in E(P)} \rel(e,v)$ and $\rel(P,w) \df \sum_{e\in E(P)} \rel(e,w)$
the sum of relevances along $P$ from the view of $v$ or $w$, respectively.
A link in a connected graph is called a \term{bridge}
if its removal destroys connectivity,
or equivalently, if it is no part of any cycle.
The relevance $\rel_G(e,v)$ is $0$ iff $e$ is \emphasis{not} a bridge.
Given a strategy profile $S$ where $G(S)$ is connected,
we define the \term{indirect cost} of player $v$ by
\begin{equation*}
  I_v(S) \df I_v(G(S)) \df
  \sum_{e\in E(S)} \rel_{G(S)}(e,v) \,\, \Pr[_{G(S)}]{\set{e}}\period
\end{equation*}
When $S$ is clear from context, 
we omit \enquote{$S$} and \enquote{$G(S)$} from the notation
and we also write $m \df \card{E(S)}$ for the number of links.
This indirect cost is the expected number of vertices
to which $v$ will lose connection when exactly one link is destroyed 
from $G(S)$ randomly and according to the probability measure given by the adversary.
For this indirect cost, we use the term \term{disconnection cost} in the following.
We define disconnection cost to be $\infty$ when $G(S)$ is not connected.
It is easily seen that if $S$ is an optimum, a NE under ULF,
a PNE under BLF, or PS under BLF, then $G(S)$ is connected.
\par
The \term{separation} $\sep(e)$ of a link $e$
is the number of ordered vertex pairs that will be separated, 
\ie pairs $(v,w)$ for which no $v$-$w$ path will exist anymore,
when $e$ is destroyed.
For a bridge $e$, 
denote $\nn(e)$ the number of vertices in the connected component of $G-e$
that has a minimum number of vertices;
note $\nn(e)\leq\sfloor{\frac{n}{2}}$.
If~$e$ is not a bridge, we define $\nn(e)\df0$.
Then $\sep(e) = 2 \nn(e) \, (n-\nn(e)) \leq n^2$
and also $\sep(e) = \sum_{v\in V} \rel(e,v)$.
If $e$ is a bridge, then $\sep(e) \geq 2 \, (n-1)$.
We can express the social cost now in terms of separation:
\begin{equation*}
  \SC(S)
  = 2 m \alpha + \sum_{v\in V} 
  \sum_{e\in E} \rel(e,v) \, \Pr{\set{e}} 
  = 2 m \alpha + \sum_{e\in E} \sep(e) \, \Pr{\set{e}}
  \period
\end{equation*}
Total building cost $2m\al$ as given here is right for BLF;
for ULF it would be $m\al$.
\par
We will consider two different adversaries.
One is called \term{simple-minded} and uses $\Pr{\set{e}} = \frac{1}{m}$ for each $e \in E$,
\ie it destroys one link uniformly at random.
The other adversary is called \term{smart} and chooses the link to destroy uniformly at random from 
$\Emax\df \setst{e\in E}{\sep(e)=\sep_{\max}}$, where $\sep_{\max}\df \max_{e\in E}\sep(e)$,
\ie it chooses uniformly at random from the set of those links
where each causes a maximum number of vertex pairs to be separated.
\par
Since optima are connected, they have social cost at least $2 \, (n-1) \, \al$.
A rough bound on the price of anarchy follows for any adversary:
\begin{equation}
  \label{eqn:rough}%
  \small%
  \frac{2 m \alpha + \sum_{e\in E} \sep(e) \, \Pr{\set{e}}}{2\,(n-1)\,\al}
  = O\parens{\frac{m \alpha + n^2 \sum_{e\in E} \Pr{\set{e}}}{n\al}}
  = O\PArens{\frac{m}{n} + \frac{n}{\al}} \period
\end{equation}

\section{Previous and Related Work}
Bilateral link formation follows a concept 
given by Myerson~\cite[p.~228]{Mye02} in a different context.
Jackson and Wolinsky~\cite{JW96} in 1996 introduced the equilibrium concept of pairwise stability
and discussed several variations of PS,
including what would later be known as PNE~\cite[p.~67]{JW96}.
Fabrikant, \etal~\cite{FLM+03} in 2003 initiated the quantitative study of 
the price of anarchy in a model that fits into the framework considered here,
as per \autoref{sec:framework}.
They considered ULF and the sum-distance model, \ie $I_v(G) = \sum_{w\in V} \dist_G(v,w)$.
Corbo and Parkes~\cite{CP05} in 2005 initiated the study of the price of anarchy under BLF
and discussed convexity;
the latter was also addressed by Calv\'{o}-Armengol and \.{I}lkili\c{c}~\cite{CI05}.
Since then, numerous results have been published;
we refer to~\cite[Sec.~4]{Kli11a} for a more comprehensive discussion.
The adversary model was invented in 2010~\cite{Kli10b},
and an $O(1)$ bound on the price of anarchy for ULF for both the simple-minded
and the smart adversary was published in 2011~\cite{Kli11a};
improved constants can be found in~\cite{Kli12a}.
Corresponding results for BLF were left to be established, 
and this is accomplished here.
\par
The adversary model addresses robustness.
This has been done before, \eg in the symmetric connections model~\cite{JW96,BS08a}
and extensions of it~\cite{BG00a,HS03,HS05}.
But all of those models show substantial differences to the adversary model,
the most prominent being that in the adversary model,
failures of different links are not independent events (but mutually exclusive events)
and that the adversary model allows the probability of failure
to depend on the built graph.
We again refer to~\cite[Sec.~4]{Kli11a} for a more comprehensive discussion.
An experimental robustness study of 
an extended version of the sum-distance model is given in~\cite{CFSK04}.
\par
Generally, independent link failures model 
the unavailability of links due to, \eg
deterioration, maintenance times,
or influences affecting the whole infrastructure or large parts of it
(\eg natural disasters).
Our adversary model, on the other hand,
models the situation when faced with an entity
that is malicious but only has limited means 
so that it can only destroy a limited number of links
-- the number being limited to~$1$ for now.

\paragraph{Our Contribution.}
We prove existence of PNE and PS graphs in the adversary model
under moderate assumptions on the adversary and $\al>0$ and $n \geq 9$.
As the main result, we prove bounds on the price of anarchy, which are tight up constants:
an $O(1)$ bound for the simple-minded adversary (if $\al > \frac{1}{2}$)
and a $\Theta\parens{1+\frac{n}{\al}}$ bound for the smart adversary (if $\al>2$), 
both times under BLF and for PNE as well as~PS.
Moreover, we prove that the simple-minded adversary induces convex cost 
and thus PNE and PS coincide.
The proof for the $O(1)$ bound uses a diameter argument,
similar to ULF, but more complicated.
The key idea is to show that long paths are \enquote{unique} in certain sense
(\autoref{prop:paths-limited} and \autoref{lem:sep-bound}) in a PS graph.

\paragraph{Open Problems.}
Extensions to other adversaries would be interesting,
\eg when $\Pr{\set{e}}$ is proportional to $\sep(e)$.
Allowing the adversary to destroy more than one link appears to provide a challenging task,
since our fundamental tool for bounding the price of anarchy,
the bridge tree (introduced in \autoref{sec:bounding-total-separation}),
appears inappropriate to capture relevant structure in that case.
\par
Convexity of cost is also left open for the smart adversary.
In~\cite{Kli12a}, non-convexity is shown, 
but only outside the set of PS graphs.
It is not known whether PS and PNE indeed diverge.
\par
The $O(1)$ bound on the price of anarchy is only claimed for $\al > \frac{1}{2}$.
In fact, it holds whenever we can strictly upper-bound the number of links in any equilibrium by $2n$,
which is the case when $\al > \frac{1}{2}$.
We also manage to go slightly below $\frac{1}{2}$, showing $O\parens{\frac{1}{\sqrt{\al}}}$ 
if $\al_0(n) \leq \al \leq \frac{1}{2}$, where $\al_0(n) \approx \frac{1}{8}$ for large $n$
(given more precisely later).
We have no matching lower bound at this time,
nor is the range of $\al < \al_0(n)$ well understood.
This should be considered in future work.

\paragraph{Technical Note.}
In order not having to introduce names for all occurring constants,
we use \enquote{$O$} and \enquote{$\Omega$} notation.
For our results, we use this notation in the following understanding:
we write
\enquote{$x=O(y)$}
if there exists a constant $c>0$ such that
$x \leq c y$.
The constant may only depend
on other constants and is in particular
independent of the non-constant quantities that constitute $x$ and $y$,
\eg parameters $n$ and $\al$.
We do not implicitly require that some quantities, \eg $n$,
have to be large.
Analogously, we write
\enquote{$x=\Omega(y)$}
if there exists a constant $c>0$ such that
$x \geq c y$.
Note that
\enquote{$O$} indicates an upper bound, 
making no statement about a lower bound;
while \enquote{$\Omega$} indicates a lower bound,
making no statement about an upper bound.
We write $x=\Theta(y)$ if $x=O(y)$ and $x=\Omega(y)$;
the constants used in the \enquote{$O$}
and the \enquote{$\Omega$} statement may be different, of course.
\par
The \enquote{$o$} notation is only used in one form,
namely $o(1)$ substituting a quantity that tends to $0$
when $n$ tends to infinity,
regardless whether other parameters are fixed or not.
Whenever we write \enquote{$o(1)$} in an expression,
it is meant as an upper bound, making no statement about a lower bound.

\section{Optima and Equilibria}
\label{sec:opt-eq}
Existence of equilibria (PNE and PS) and structural results for optima and equilibria 
under moderate assumptions on the adversary (namely inducing anonymous cost,
which is fulfilled by our simple-minded and smart one)
follow easily from what has been done for ULF in~\cite{Kli11a}.
We give the details in \autoref{app:opt-eq} (Appendix) and merely provide a summary here.
If $\al \leq n-1$ then the cycle is optimal with social cost $2n\al$.
If $\al \geq n-1$ then the star is optimal with social cost $2\,(n-1)\,(\al+1) \geq 2n\al$.
In any case, optimal social cost is $\Om(n\al)$; it is $\Theta(n\al)$ if $\al=\Om(1)$.
Star and cycle also occur as PNE for the two adversaries:
the star for $\al > 2-\frac{1}{n-1}$ and the cycle for $\al \leq \frac{1}{2} \sfloor{\frac{n-1}{2}}$.
If $n \geq 9$, then those regions overlap, so we have guaranteed existence 
of PNE for $n \geq 9$.

\section{Simple-Minded Adversary}
We have for the simple-minded adversary cost and social cost:
\begin{align}
  \nonumber
  C_v(S) & = \ones{S_v}\,\alpha + \frac{1}{m} \sum_{e\in E} \rel(e,v) 
  = \ones{S_v}\,\alpha + \frac{1}{m} \Rel(v) \\
  \label{eqn:sc}
  \SC(S) &= 2m\alpha + \frac{1}{m} \sum_{v\in V} \Rel(v) 
  = 2m\alpha + \frac{1}{m} \sum_{e\in E} \sep(e) 
\end{align}

\subsection{Convexity of Cost}
Convexity is interesting since it implies PNE and PS being equivalent.
The proof of the following proposition is given in \autoref{app:becoming-bridges} (Appendix).
\begin{thmprop}
  \label{prop:becoming-bridges}%
  Let $G=(V,E)$ be a connected graph,
  $v\in V$, $e=\set{v,w} \in E$, and $F\subseteq E\setminus\set{e}$ a set of links,
  each incident with $v$,
  so that $G'\df G-F-e$ is still connected.
  Let $B_1$ be those links that are non-bridges in $G$ but bridges in $G-e$.
  Let $B_2$ be those links that are non-bridges in $G-F$ but bridges in $G-F-e$.
  Then $B_1\subseteq B_2$.
\end{thmprop}
\begin{thmthm}
  \label{lem:simple-minded-convex}%
  The simple-minded adversary induces convex cost.
\end{thmthm}
\begin{proof}
  Let $v\in V$ and $w_1,\hdots,w_k\in V$
  and $S$ be a strategy profile.
  We show~\eqref{eqn:def-convex-2} 
  proceeding by induction on $k$.
  The case $k=1$ is clear.
  Let $k>1$ and set $S'\df S-vw_1-\hdots-vw_{k-1}$.
  We show that switching from $S'$ to $S'-vw_k$
  increases disconnection cost for $v$ at least as much as switching from $S$ to $S-vw_k$, \ie
  \begin{equation}
    \label{eqn:convexity-induction-step}%
    I_v(S'-vw_{k}) - I_v(S')
    \geq I_v(S-vw_k) - I_v(S)\period
  \end{equation}
  Denoting $R_1' \df R_{S'-vw_k}(v)$
  and $R_2' \df R_{S'}(v)$
  and $R_1 \df R_{S-vw_k}(v)$
  and $R_2 \df R_{S}(v)$,
  a routine calculation using $R_1' \geq R_1$ shows that to this end,
  it suffices to prove $R_1' - R_2' \geq R_1 - R_2$ 
  (have to consider dividing by the number of links).
  \par
  As removing $\set{v,w_k}$, relevance of zero or more links 
  changes from $0$ to a positive value;
  these are precisely those links which become bridges by the removal
  and which were no bridges before.
  No relevance is reduced by a removal.
  \par
  Let $B_1$ be all those links that become bridges by the switch from $S$ to $S-vw_k$,
  and let $B_2$ those that become bridges by the switch from $S'$ to $S'-vw_k$.
  Then $B_1\subseteq B_2$ by \autoref{prop:becoming-bridges}.
  The increase in relevance from $0$ to a positive value for $e\in B_1$
  given $S'$ is at least as high as when given $S$.
  In other words,
  while $\set{v,w_1},\hdots,\set{v,w_{k-1}}$ are removed,
  the effect of all links in $B_1$ becoming bridges is saved 
  until the removal of $\set{v,w_k}$.
  We have shown that $R_1' - R_2' \geq R_1 - R_2$ and thus~\eqref{eqn:convexity-induction-step}.
  \par
  The proof is concluded by the following standard calculation:
  \begin{align*}
    &\quad\: I_v(S-vw_1-\hdots-vw_k)-I_v(S)\\
    &=I_v(S'-vw_k)-I_v(S')+I_v(S')-I_v(S)\\
    &\geq I_v(S-vw_k)-I_v(S)+I_v(S')-I_v(S) & \text{by~\eqref{eqn:convexity-induction-step}}\\
    &\geq I_v(S-vw_k)-I_v(S)
    + \sum_{i=1}^{k-1} \Parens{I_v(S-vw_i) - I_v(S)} & \text{by induction} \\
    & = \sum_{i=1}^{k} \Parens{I_v(S-vw_i) - I_v(S)} \period
    \tag*{\qedhere}
  \end{align*}
\end{proof}

\subsection{Bounding Total Building Cost}
As a first step towards our bound on the price of anarchy,
we bound the number of links and hence total building cost in a PS graph.
A \term{chord} is a link between two vertices on the same cycle.
The following proposition was proved with $3n$ in~\cite{Kli11a} and improved to $2n$ in~\cite{Kli12a}.
\begin{thmprop}
  \label{prop:chord-free-general}%
  A chord-free graph on $n$ vertices has strictly less than $2n$ links.
\end{thmprop}
\begin{thmprop}
  \label{prop:number-links}%
  Let a pairwise stable graph $G$ be given.
  \begin{enumerate}[label=(\roman*)]
    \item\label{prop:chord-free-bilateral-simple:i}%
      If $\alpha>\frac{1}{2}$, then $G$ is chord-free
      and hence has $<2n=O(n)$ links.
    \item\label{prop:chord-free-bilateral-simple:ii}%
      In general, $G$ is chord-free
      (with $O(n)$ links)
      or has at most $\frac{n}{\sqrt{2\alpha}}+1$ links.
  \end{enumerate}
\end{thmprop}
\begin{proof}
  If $G$ is bridgeless, selling a chord is beneficial since
  disconnection cost $0$ is maintained.
  So for both parts we assume that $G$ contains bridges.
  \par
  \ref{prop:chord-free-bilateral-simple:i}
  The impairment in disconnection cost for a player $v$ of selling a chord 
  is only due to the change in the denominator of the disconnection cost
  and is precisely $\frac{1}{m\,(m-1)}\,\Rel(v)$,
  which is upper-bounded by $\frac{1}{2}$
  since $\Rel(v)\leq\frac{n\,(n-1)}{2}$~\cite[Prop.~8.3]{Kli11a}.
  Hence if $\alpha$ is larger than that,
  there is an incentive to sell the chord.
  \par
  \ref{prop:chord-free-bilateral-simple:ii}
  Let $G$ possess a chord.
  This means that any of its two endpoints, say $v$, 
  deems it being no impairment
  to pay $\alpha$ for this link,
  hence $\frac{1}{m\,(m-1)}\,\Rel(v) \geq \alpha$.
  It follows 
  \begin{equation*}
    \frac{n^2}{2} \geq \frac{n\,(n-1)}{2} \geq \Rel(v) 
    \geq m\,(m-1)\,\alpha \geq (m-1)^2\,\alpha
    \enspace,
  \end{equation*}
  hence $\frac{n}{\sqrt{2\alpha}}+1\geq m$.
\end{proof}

\subsection{Bounding Total Separation}
\label{sec:bounding-total-separation}
Since the adversary only destroys one link,
it can only cause damage if it chooses a bridge.
Therefore, the bridge structure of the built graph is important.
Let $G=(V,E)$ be any connected graph.
We call $K \subseteq V$ a \term{bridgeless connected component}
or just \term{component} if it is inclusion-maximal
under the condition that the induced subgraph $G[K]$ does not contain any bridges of $G[K]$,
or equivalently does not contain any bridges of $G$.
It is easy to see that the set of all components forms a partition of $V$.
The \term{bridge tree} $\tG=(\tV,\tE)$ of $G$ is the following graph:
\begin{align*}
  \tV&\df\setst{K\subseteq V}{\text{$K$ is a component}}\\
  \tE&\df \setst{\set{K,K'} \in \textstyle{\tV \choose 2}}%
  {\exists v\in K, w\in K' \holds \set{v,w}\in E}
\end{align*}
Another way of thinking of the bridge tree is that we successively contract each cycle
to a new vertex which is adjacent to all vertices that had a neighbor in the contracted cycle. 
If for each $v\in V$ we denote $\kappa(v)$ the unique component with $v \in \kappa(v)$, then
it is easy to see that $\set{v,w} \mapsto \set{\kappa(v), \kappa(w)}$
is a bijection between the bridges of $G$ and the links of $\tG$.
In the game theoretic situation, we look at the bridge tree of the built graph $G(S)$.
When considering the effect of building additional links,
we may treat vertices of the bridge tree as players.
This is justified since links inside components have $0$ relevance.
Hence for a strategy profile $S$ and $\set{K,K'} \in {\tV \choose 2}$ 
the effect in disconnection cost of a new link between a player from $K$
and a player from $K'$ is specific to the pair $\set{K,K'}$
and not to the particular players.
Note also that a link $e=\set{v,w}$ with $v,w$ in the same component
not only has $\rel(e,x)=0$ for all $x\in V$, but also $\sep(e)=0$.
\par
Knowing this \emphasis{we will only work with the bridge tree during the rest of this section},
so all vertices, links, and paths are in the bridge tree.
We also make the convention that
whenever we speak of the number of vertices in a subgraph $T$ of the bridge tree,
we count $\card{v}$ for each vertex $v\in V(T)$,
\ie we count the vertices that are in the respective components
(illustration in \autoref{app:bridge-tree}, Appendix).
We call a bridge tree PS if it stems from a PS graph.
\par
The proof of the following remark is straightforward:
\begin{thmrem}
  \label{rem:bypass}%
  Let $P=(v,\hdots,w)$ be a path (reminder: in the bridge tree).
  The benefit in disconnection cost of \term{bypassing} $P$ can be lower-bounded:
  \begin{equation*}
	I_v(G) - I_v(G+\set{v,w}) \geq \frac{\sum_{e\in E(P)} \rel(e,v)}{m+1} =  \frac{\rel(P,v)}{m+1} \period
  \end{equation*}
\end{thmrem}
The following bound on total separation holds for all connected graphs:
\begin{thmlem}[\protect{\cite[Cor.~8.12]{Kli11a}}]
  \label{lem:diam-bound}
  For each $v$ we have $R(v) \leq (n-1) \, \diam(\tG)$, hence
  \begin{equation*}
	\sum_{e\in E} \sep(e) < n^2 \diam(\tG) \period
  \end{equation*}
\end{thmlem}
Since $m \geq n-1$, it follows that total disconnection cost is bounded by 
\begin{equation*}
  \frac{n^2}{m} \diam(\tG) = n \diam(\tG) \, (1+ o(1)) \comma
\end{equation*}
where $o(1)$ is for $n \conv \infty$.
Under ULF, if $S$ is a NE then $\diam(\tG(S)) = O(\al)$~\cite[Lem.~8.13]{Kli11a},
by which an $O(1)$ bound is obtained.
If we could prove the same for BLF and PS, we would be done.
However, an example shows that there is no hope for this;
the diameter of a PS bridge tree can be $\Om(\sqrt{n})$ (\autoref{app:example-diameter}, Appendix).
The trick to handle this situation is to prove that essentially, 
there is only \emphasis{one} such long path.
Note that in non-PS graphs, total separation can be $\Om(n^3)$,
\eg if the graph is a path.
\par
For a moment, orient the links of the bridge tree according to the rule
that $e=\set{v,w}$ is oriented $(v,w)$ if $\rel(e,v) \leq \rel(e,w)$,
\ie links point in the direction of fewer vertices, ties broken arbitrarily.
Due to cycle-freeness, there is a vertex $u$ having only out-links;
and since each vertex has at most one in-link, this vertex $u$ is unique.
We consider the bridge tree rooted at $u$
and discard the orientation, it has served its purpose.
Denote $\lev(v)$ the level of $v$ in the rooted bridge tree, 
in particular $\lev(u)=0$ and $\lev(v)=1$ for $v \in N(u)$.
Each of the subtrees rooted at some $v \in N(u)$ is called a \term{branch} of the bridge tree.
\par
A path $P=(v,\hdots,w)$ is called \term{$v$-limited} if $\rel(P,v) \leq 2 n \al$.
It is called \term{$w$-limited} if $\rel(P,w) \leq 2 n \al$.
A path is called a \term{branch path} if it traverses each level at most once,
\ie it runs straight up or down and in particular stays within one branch plus $u$.
A branch path $P=(v,\hdots,w)$ with $\lev(v) < \lev(w)$ 
is called \term{outward-limited} if it is $v$-limited
and is called \term{inward-limited} if it is $w$-limited.
Denote the constant $c \df 4$
and also $\al_0(n) \df \frac{1}{8} \Parens{1+\frac{1}{n-1}}^2 \leq \frac{9}{32}$.
We have:
\begin{thmprop}
  \label{prop:paths-limited}
	Let $\al \geq \al_0(n)$.
  \begin{enumerate}[label=(\roman*)]
    \item\label{prop:paths-limited:i}
      In a PS bridge tree, each path $P=(v,\hdots,w)$ is $v$-limited or $w$-limited (or both).
      In particular, each branch path is inward-limited or outward-limited (or both).
    \item\label{prop:paths-limited:ii}
      An inward-limited path has length at most $c\al=O(\al)$.
    \item\label{prop:paths-limited:iii}
      In a PS bridge tree,
      no more than one branch contains a path not being inward-limited.
  \end{enumerate}
\end{thmprop}
\begin{proof}
  \ref{prop:paths-limited:i}~Assume the contrary, \ie $\rel(P,v) > 2n\al$ and $\rel(P,w) > 2n\al$.
  Then for each $x \in e \df \set{v,w}$ by \autoref{rem:bypass} and \autoref{prop:number-links}
	we have, using $\al \geq \al_0(n)$ for the second case (\ie when $\al \leq \frac{1}{2}$),
  \begin{equation*}
	I_x(G) - I_x(G+e) > 
	\frac{2n\al}{m+1} \geq
	\begin{cases}
	  \frac{2n\al}{2n} = \al \\
	  \frac{2n\al}{\frac{n}{\sqrt{2\al}} + 2} \geq \al
	\end{cases}\comma
  \end{equation*}
  a contradiction to PS.
  \par
  \ref{prop:paths-limited:ii}~Let $P=(v,\hdots,w)$ be inward-limited with $\lev(v)<\lev(w)$
  and of length~$\ell$.
  Let $N \df n - n'$ with $n'$ the number of vertices in $P$'s branch;
  then by construction $N \geq \frac{n}{2}$.
  It follows $2n\al\geq \rel(P,w) \geq N \ell \geq \frac{n}{2} \ell$, hence $4\al \geq \ell$.
  (This type of argument is the basis for $\diam(\tG(S)) = O(\al)$ 
  for ULF in \cite[Lem.~8.13]{Kli11a}.)
  \par
  \ref{prop:paths-limited:iii}~Assuming the contrary, let $(v,\hdots,w)$ and $(x,\hdots,y)$
  with $\lev(v)<\lev(w)$ and $\lev(x) < \lev(y)$ from different branches and both not inward-limited.
  Then $(w,\hdots,v,\hdots,u,\hdots,x,\hdots,y)$ is not $w$-limited nor $y$-limited,
  contradicting~\ref{prop:paths-limited:i}.
\end{proof}
\begin{thmlem}
  \label{lem:sep-bound}
  In a PS graph with $\al\geq\al_0(n)$, total separation is bounded by  $2\,(2+c)\,n^2 \al$.
\end{thmlem}
\begin{proof}
  Let $P=(u,\hdots,w)$ be a path of maximum length $\ell$ among all path starting at $u$.
  If $\ell \leq c\al$, we are done by \autoref{lem:diam-bound}.
  Otherwise, if $\ell > c\al$, by \autoref{prop:paths-limited}\ref{prop:paths-limited:ii},
  $P$ is not inward-limited.
  By \autoref{prop:paths-limited}\ref{prop:paths-limited:iii},
  branch paths in all other branches are inward-limited and hence of length at most $c\al$.
  \par
  We consider paths in $P$'s branch of the form $Q=(x,\hdots,y)$ with $x \in V(P)$ and $y\neq w$ a leaf.
  Denote $k\df c\al + 1$ if $c\al$ is integer and $k \df \sceil{c\al}$ otherwise.
  Call the final $k$ vertices of $P$ (before it ends with $w$) its \term{lower part}
  and the other vertices (starting with $u$) its \term{upper part}.
  If $Q$ attaches in the lower part of $P$, \ie $x$ is in the lower part,
  then $\len{Q} \leq k-1 \leq c\al$ since $P$ has maximum length.
  If, on the other hand, $Q$ attaches to the upper part of $P$,
  then $A \df (w,\hdots,x)$ has length at least $k > c \al$.
  Denote $B \df (w,\hdots,x,\hdots,y)$.
  If $Q$ is not inward-limited, then $B$ is not $y$-limited, 
  by \autoref{prop:paths-limited}\ref{prop:paths-limited:i} it is hence $w$-limited.
  In particular, $A$ is $w$-limited, hence inward-limited 
  and thus of length at most $c\al$, a contradiction.
  It follows that $Q$ is inward-limited and so $\len{Q} \leq c\al$.
  \par
  We bound sums of separation values.
  Since $P$ is not inward-limited, it is outward-limited, \ie $u$-limited.
  Denote $P=(u=w_0,\hdots,w_\ell=w)$.
  It follows 
  \begin{equation}
    \label{eqn:sep-long-path}
    \begin{split}
      \sum_{e\in E(P)} \sep(e) 
      & = 2 \sum_{e=\set{w_i,w_{i+1}}\in E(P)} \rel(e,w_i) \, \rel(e,{w_{i+1}}) \\
      & \leq 2 \sum_{e=\set{w_i,w_{i+1}}\in E(P)} \rel(e,u) \, n
      = 2 n \, \rel(P,u) \leq 4 n^2 \al \period
    \end{split}
  \end{equation}
  If we bypass $P$, \ie add $\set{u,w}$,
  then separation values for all $e \in E' \df E \setminus E(P)$ are maintained.
  Hence we can compute $\sum_{e \in E'} \sep(e)$ as if $P$ had been bypassed.
  But with the bypass, the diameter of the bridge tree is $2 c \al$ by what we have established above.
  By \autoref{lem:diam-bound} it follows $\sum_{e \in E'} \sep(e) \leq n^2 \cdot 2c\al$.
  Together with \eqref{eqn:sep-long-path}, the bound on total separation follows.
\end{proof}
\begin{thmthm}
  \label{thm:O1}
  Price of anarchy for the simple-minded adversary with BLF and PNE or PS is 
  $1 + (2+c) \, (1+o(1)) = 1 + 6 \cdot (1+o(1)) = O(1)$ if $\al > \frac{1}{2}$
  and 
  $\max\Set{1,\,\frac{1}{\sqrt{8\al}}} + (2+c) \, (1+o(1)) = O\parens{\frac{1}{\sqrt{\al}}}$ if $\al_0(n) \leq \al \leq \frac{1}{2}$,
  where $o(1)$ is for $n \rightarrow \infty$.
\end{thmthm}
\begin{proof}
  Recall the lower bound of $2n\al$ for an optimum.
  By \autoref{prop:number-links},
  for total building cost we have the ratio of equilibrium to optimum bounded by
  $\frac{2n\al}{2n\al} = 1$ if $\al > \frac{1}{2}$
  and otherwise, if $\al \leq \frac{1}{2}$, the same bound or:
  \begin{equation*}
    \frac{\frac{n}{\sqrt{2\al}}\al+\al}{2n\al} = \frac{1}{\sqrt{8\al}} + o(1) \period
  \end{equation*}
  For total disconnection cost, using $m \geq n-1$, we have the bound:
  \begin{equation*}
    \frac{\frac{2\,(2+c)\,n^2\al}{m}}{2n\al} 
    = \frac{(2+c)\,n}{m}
    \leq (2+c) \, \frac{n}{n-1}
    = (2+c) \, (1+o(1)) \period\qedhere
  \end{equation*}
\end{proof}

\section{Smart Adversary}
Recall that the smart adversary destroys one link uniformly at random from
$\Emax = \setst{e\in E}{\sep(e)=\sep_{\max}}$.
We call the links in $\Emax$ the \term{critical} links.
\begin{thmrem}
  \label{rem:chord-free-smart}%
  A PS graph is chord-free for the smart adversary.
\end{thmrem}
\begin{proof}
  Removing a chord does not change the relevance of any link,
  nor does it change $\sep_{\max} \, (=\max_e \sep(e))$, hence it does not change $\Emax$.
  A player will hence always opt to remove a chord and so avoid the expense of $\al$.
\end{proof}
Given chord-freeness, we have the bound of $2n=O(n)$ on the links in a PS graph.
As seen by \eqref{eqn:rough},
this implies a bound of $O\parens{1+\frac{n}{\alpha}}$
on the price of anarchy.
This bound holds for any adversary provided that $m=O(n)$ in all equilibria.
It is tight for the smart adversary and BLF as shown by the following example.
\begin{SCfigure}
  \caption{%
  \label{fig:three-stars}%
  Three stars of sizes $n_0$, $n_0-1$, and $n_0-2$; here $n_0=5$.
  The $n_0$ players in the star around $u_1$ would like to put 
  the one critical link $\set{u_0,u_1}$ on a cycle,
  if $\alpha<n_0$.
  Building, \eg $\set{u_1,u_2}$ would reduce their disconnection cost from $n-n_0$ to $n_0-2$,
  meaning an improvement of $n_0$.
  But no player from the stars around $u_2$ or $u_3$ is willing to cooperate.}
  \begin{tikzpicture}[%
    every node/.style=player,%
    node distance=1cm,%
    ]
    \node (u0) {$u_0$};
    \node (u1) [above right = 2.5cm of u0] {$u_1$};
    \node (u11) [below = of u1] { };
    \node (u12) [right = of u1] { };
    \node (u13) [above = of u1] { };
    \node (u14) [left = of u1] { };
    \node (u2) [left = 2cm of u0] {$u_2$};
    \node (u21) [below = of u2] { };
    \node (u22) [left = of u2] { };
    \node (u23) [above = of u2] { };
    \node (u3) [below = 2cm of u0] {$u_3$};
    \node (u31) [below right = of u3] { };
    \node (u32) [below left = of u3] { };
    \path (u0) edge[dashed] (u1)
    (u0) edge         (u2)
    (u0) edge         (u3);
    \path (u1) edge (u11)
    (u1) edge (u12)
    (u1) edge (u13)
    (u1) edge (u14);
    \path (u2) edge (u21)
    (u2) edge (u22)
    (u2) edge (u23);
    \path (u3) edge (u31)
    (u3) edge (u32);
  \end{tikzpicture}
\end{SCfigure}
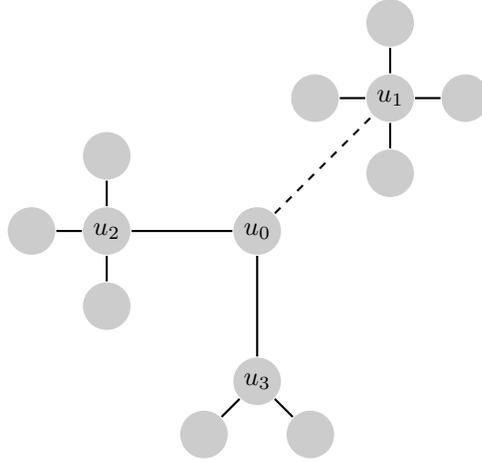%
\begin{thmthm}
  \label{thm:bilateral-smart-lower}%
  Let $n \geq 9$ and $\al > 2$, 
  then the price of anarchy for PNE and PS with the smart adversary 
  is $\Omega(1 + \frac{n}{\alpha})$.
\end{thmthm}
\begin{proof}
  Assume there is an integer $n_0\geq 3$ such that $n=3n_0-2$;
  the general case can be proved in a similar way.
  Consider three stars $U_i$, $i=1,2,3$ with center vertices $u_i$, $i=1,2,3$,
  and $n_0$, $n_0-1$, and $n_0-2$ vertices, respectively.
  Connect the stars via an additional vertex $u_0$ 
  and additional links $\set{u_0,u_i}$, $i=1,2,3$.
  This construction uses $3n_0-2$ vertices.
  Then $e_0=\set{u_0,u_1}$ is the only critical link
  and $n_0=\Theta(n)$,
  namely slightly more than $\frac{1}{3}n$.
  We have a total disconnection cost of $2 \nn(e_0) \, (n-\nn(e_0))
  = 2 n_0 \, (n-n_0) = \Omega(n^2)$,
	and so a social cost of $\Omega(n\al + n^2)$.
  It follows a ratio to the optimum of $\Omega(1+\frac{n}{\alpha})$.
  We are left to show that
  this graph is a PNE,
  which implies PS.
  Clearly, no link can be sold,
  since that would make the graph disconnected.
  Therefore we only have to ensure 
  that no link can be added that would be beneficial for one endpoint
  and at least no impairment for the other one,
  \ie we have to show~\eqref{eqn:pne}.
  \par
  A link $e$ that improves disconnection cost for some player 
  has to put $\set{u_0,u_1}$ on a cycle.
  If $e$ connects a vertex in $U_1$ with a vertex in $U_2$,
  then $\set{u_0,u_3}$ will become critical.
  For a vertex in $U_2$,
  this reduces disconnection cost from $n_0$ to $n_0-2$.
  So, since $\alpha> 2$, no vertex in $U_2$ agrees to build such a link.
  \par
  A similar situation holds if $e$ connects a vertex in $U_1$ with a vertex in $U_3+u_0$.
  It will result in $\set{u_0,u_2}$ becoming critical.
  For a vertex in $U_3+u_0$,
  this reduces disconnection cost from $n_0$ to $n_0-1$.
  So, since $\alpha> 1$, no vertex in $U_3+u_0$ agrees to build such a link.
\end{proof}
\par
If we consider $\alpha>2$ a constant, this theorem gives a lower bound of $\Omega(n)$.
For $\alpha=\Omega(1)$ this is by~\eqref{eqn:rough}
the worst that can happen for \emphasis{any} adversary in this model (note $m \leq n^2$).
This is particularly noteworthy since for ULF and NE,
the smart adversary has a price of anarchy of $O(1)$~\cite[Thm.~9.8]{Kli11a}.

\section*{Acknowledgement} 
I thank Mourad El Ouali for proofreading.
I thank the anonymous referees for helpful comments
and pointing out an error for the case $\al \leq \frac{1}{2}$
in the submitted version of this work.
I also thank the organizers of SAGT 2013.

\pagebreak

\newpage
\appendix
{\noindent\Large\textbf{Appendix}}

\section{Optima and Equilibria}
\label{app:opt-eq}
For optima, we repeat a proof from~\cite[Prop.~7,1]{Kli11a} for ULF with minimal modifications,
which adapt it to BLF.
This result is for a general adversary,
\ie one that destroys one link according to some probability measure $\P$
on the links of the built graph.
\begin{thmprop}
  \label{prop:uni-opt}%
  An optimum has social cost $\Om(n\alpha)$, and $\Theta(n\alpha)$ if $\al=\Om(1)$.
  More precisely:
  \begin{enumerate}[label=(\roman*)]
    \item If $\alpha\leq n-1$, the cycle is an optimum;
      it has social cost $2n\alpha$.
    \item If $\alpha\geq n-1$, a star is an optimum;
      it has social cost $2 \, (n-1)\,(\alpha+1)$.
  \end{enumerate}
\end{thmprop}
\begin{proof}
  An optimum can only be the cycle or a tree,
  because any graph containing a cycle has already the building cost $n\alpha$ of the cycle,
  and the cycle has optimal disconnection cost.
  So an optimum is either the cycle, or it is cycle-free.
  Let $T$ be any tree.
  We have its disconnection cost:
  \begin{align*}
    \MoveEqLeft
    \sum_{e\in E(T)} \sep(e) \, \Pr{\set{e}}
    =2\sum_{e\in E(T)}\nn(e)\,(n-\nn(e)) \, \Pr{\set{e}}\\
    & \geq2\cdot1\,(n-1)\sum_{e\in E(T)}\Pr{\set{e}}
    =2\,(n-1)\period
  \end{align*}
  We use that the function $x\mapsto2x\,(n-x)$ is (strictly) increasing 
  on $[1,\frac{n}{2}]$ for the lower bound.
  We conclude that the social cost of a tree is at least
  \begin{equation}
    \label{eqn:lower-bound-tree}
    2 \, (n-1)\,\alpha+2\,(n-1)= 2 \, (n-1)\,(\alpha+1)\period
  \end{equation}
  Social cost of the cycle is $2n\alpha$.
  So if $\alpha\leq n-1$, the cycle is better or as good as any tree,
  hence it is an optimum.
  If $\alpha > n-1$, then we look for a good tree.
  A~star has social cost $2 \, (n-1)\,(\alpha+1)$, 
  which matches the lower bound~\eqref{eqn:lower-bound-tree}
  and is hence optimal (and better than the cycle).
\end{proof}
Note that, indeed, for $\al=n-1$, cycle and star co-exist as optima with social cost 
$2n\al=2n\,(n-1)=2\,(\al+1)\,(n-1)$.
\par
For equilibria, we also build on the results from~\cite{Kli11a}.
We call a NE $S$ a \term{maximal Nash equilibrium} (MaxNE)
if $C_v(S+vw_1+\hdots+vw_k) > C_v(S)$
for all $k\in\setn{n}$ and $\set{v,w_1},\hdots,\set{v,w_k}\in \nE(S)$ and $v\in V$.
That is, we exclude the possibility 
that a player can buy additional links so
that the gain in her indirect cost and the additional building cost
nullify each other.
We make the convention that whenever we speak of NE or MaxNE, we mean that relative to ULF.
Whenever we speak of PNE or PS, we mean that relative to BLF.
\par
We also need the notion of \term{anonymous cost}.
We call indirect cost \term{anonymous} if
for each built graph $G=(V,E)$
and each graph automorphism $\fn{\phi}{V\map V}$ of $G$,
we have $I_v(G)=I_{\phi(v)}(G)$ for all $v\in V$.
In other words, anonymity of indirect cost means that
$I_v(G)$ does not depend on $v$'s identity,
but only on $v$'s position in the built graph $G$.
This is of importance in particular if $G$ has symmetry.
For instance, if $G$ is a cycle and indirect cost is anonymous,
then all vertices experience the same indirect cost.
If $G$ is a path and indirect cost is anonymous,
then both endpoints experience the same indirect cost.
It is easy to see that both our adversaries, the simple-minded and the smart one,
induce anonymous cost.
\par
For ULF, we cite the following two propositions from~\cite{Kli11a} without proofs.
Strategy profiles are described by directed graphs.
This is straightforward: a strategy profile $S\in\OI^{n \times n}$
can be interpreted as the adjacency matrix for a directed graph on $V=\setn{n}$,
and vice versa.
A directed link $(v,w)$ in that graph means a request by $v$ for the link $\set{v,w}$.
Under ULF, such a request results in the link being built,
and then we call $v$ the \term{owner} of the link.
Since we restrict to essential strategy profiles, 
under ULF there is always exactly one owner per link.
\begin{thmprop}[\protect{\cite[Prop.~7.3]{Kli11a}}]
  \label{prop:uni-ne-star}%
  Let $S$ be a star
  with links pointing outward 
  (\ie the built graph is a star and all links are paid for by the center vertex).
  \begin{enumerate}[label=(\roman*)]
    \item\label{prop:uni-ne-star:i}%
      If $\alpha\geq n-1$, then $S$ is a NE.
    \item\label{prop:uni-ne-star:ii}%
      If $\alpha\geq 2-\frac{1}{n-1}$, 
      then $S$ is a NE if disconnection cost is anonymous.
  \end{enumerate}
  In both cases, strict inequality implies a MaxNE.
\end{thmprop}
\begin{thmprop}[\protect{\cite[Prop.~7.4]{Kli11a}}]
  \label{prop:uni-ne-cycle}%
  Let $S$ be a cycle
  with all links pointing in the same direction
  (either all clockwise or all counter-clockwise).
  \begin{enumerate}[label=(\roman*)]
    \item\label{prop:uni-ne-cycle:i}%
      If $\alpha\leq 1$, then $S$ is a MaxNE.
    \item\label{prop:uni-ne-cycle:ii}%
      If $\alpha\leq \frac{1}{2}\sfloor{\frac{n-1}{2}}$, 
      then $S$ is a MaxNE if disconnection cost is anonymous.
  \end{enumerate}
\end{thmprop}
\begin{thmcor}
  \label{cor:uni-ne}%
  For anonymous disconnection cost,
  MaxNE exist for all ranges of $\alpha>0$ provided that $n\geq 9$.
\end{thmcor}
\begin{proof}
  Just note that $2-\frac{1}{n-1} < 2\leq\frac{1}{2}\sfloor{\frac{n-1}{2}}$ for $n\geq 9$
  and invoke the preceding propositions.
\end{proof}
\par
Certain classes of simple-structured MaxNE under ULF are PNE under~BLF.
The following two propositions hold in a more general setting,
not limited to the adversary model.
Let $S$ be a strategy profile.
Define $S^\Bi$ by
$S^\Bi_{vw} \df \min\set{1,\, S_{vw}+S_{wv}}$ for all $v,w\in V$.
Then $G^\Uni(S)=G^\Bi(S^\Bi)$.
In other words, forming $S^\Bi$ means adding to $S$
the necessary requests so that for BLF the same graph is built as we have for ULF.
\begin{thmprop}
  Let $S$ be a MaxNE under ULF with $G\df G^\Uni(S)$ being a cycle
  and using anonymous indirect cost. 
  Then $S^\Bi$ is a PNE under~BLF.
\end{thmprop}
\begin{proof}
  By the definition of MaxNE,
  any additional link is an impairment for the buyer.
  So the premise of~\eqref{eqn:pne} 
  is never true, \ie all absent links are justified.
  \par
  New links cannot be formed unilaterally.
  We are hence left to show that each link is wanted by both endpoints,
  \ie none of the endpoints can improve her individual cost by deleting the link.
  Let $v$ be the owner of a link $\set{v,w}$ under ULF.
  Since we have a NE there, 
  $v$ cannot improve her individual cost by selling this link.
  Selling the link means that $v$ would be at the end of the path $G-\set{v,w}$.
  By anonymity of indirect cost we conclude:
  it is worth or at least no impairment paying $\alpha$ for not being at the end of the path
  that results from $G$ by deletion of one link.
  Therefore, both of each two neighboring vertices 
  maintain their requests in $S^\Bi$ for having a link between them.
\end{proof}
\begin{thmprop}
  Let $S$ be a MaxNE under ULF with $G\df G^\Uni(S)$ being a tree.
  Let indirect cost assign $\infty$ to a disconnected graph.
  Then $S^\Bi$ is a PNE under~BLF.
\end{thmprop}
\begin{proof}
  As in the previous proposition,
  \eqref{eqn:pne} follows from the properties
  of a MaxNE.
  So we are left to consider removals.
  Since the final graph is a tree,
  removal of any link would make it disconnected, yielding indirect cost $\infty$.
  Hence no player wishes to remove a link.
\end{proof}
\par
We turn to the adversary model.
It follows from the two previous propositions that the equilibrium existence results from
\autoref{prop:uni-ne-star}, \autoref{prop:uni-ne-cycle} and \autoref{cor:uni-ne}
carry over from MaxNE to PNE. We so have:
\begin{thmcor}
  For anonymous disconnection cost (which includes the simple-minded and the smart adversary)
  PNE, and hence also PS graphs, exist under BLF for $\alpha>0$ and $n\geq 9$.
\end{thmcor}

\section{Proof of \autoref{prop:becoming-bridges}}
\label{app:becoming-bridges}
\begin{thmrem}
  \label{rem:one-cycle}
  Let $G=(V,E)$ be a graph and $e=\set{v,w}\in E$ a non-bridge
  and $C$ \emphasis{any} cycle with $e\in E(C)$.
  Then all bridges of $G-e$ that are non-bridges in~$G$,
  are in~$E(C)$.
\end{thmrem}
\begin{proof}
  Let $f$ be a non-bridge in $G$ and a bridge in $G-e$.
  Then $G-e$ consists of two subgraphs $G_1$ and $G_2$ that are connected only by $f$.
  Since $f$ was no bridge before $e$ was removed,
  $e$ must also connect $G_1$ with $G_2$.
  Moreover, there are no other links between $G_1$ and $G_2$.
  It follows that any cycle that contains $e$ also contains $f$.
\end{proof}
\begin{proof}[Proof of \autoref{prop:becoming-bridges}]
  Since $G'$ is connected,
  there is a path $(v,e_1,v_1,\hdots,w)$ in~$G'$.
  Then the cycle $C\df(v,\hdots,w,e,v)$ is in $G-F$.
  By \autoref{rem:one-cycle}, we have $B_1\subseteq E(C)$.
  Hence all links in $B_1$ are on a cycle that is not destroyed by removal of $F$,
  so no link in $B_1$ is made a bridge by removal of $F$.
  It follows $B_1\subseteq B_2$.
\end{proof}

\section{Illustration for the Bridge Tree}
\label{app:bridge-tree}
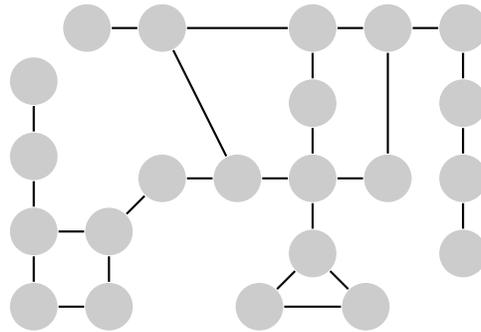
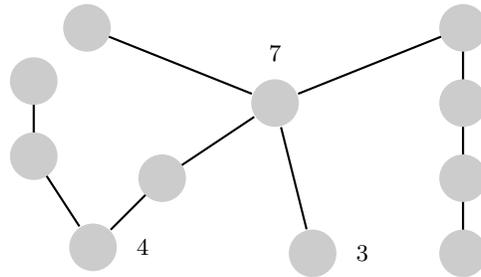
\begin{figure}[H]
\newcommand{\skel}{%
	\node (a1) { };
	\node (a2) [below = of a1] { };
	\node (a3) [left = of a2] { };
	\node (a4) [above = of a3] { };
	\node (b1) [above right = of a1] { };
	\node (b2) [above = 2cm of b1] { };
	\node (b3) [right = 2cm of b2] { };
	\node (b4) [right = of b3] { };
	\node (b5) [below = 2cm of b4] { };
	\node (b6) [left = of b5] { };
	\node (b7) [left = of b6] { };
	\node (b8) [below = of b3] { };
	\node (c1) [left = of b2] { };
	\node (d1) [above = of a4] { };
	\node (d2) [above = of d1] { };
	\node (e1) [below = of b6] { };
	\node (e2) [below left = of e1] { };
	\node (e3) [below right = of e1] { };
	\node (f1) [right = of b4] { };
	\node (f2) [below = of f1] { };
	\node (f3) [below = of f2] { };
	\node (f4) [below = of f3] { };}
\centering
\subfloat[]%
[\label{fig:bridge-tree:1}A graph $G$.]%
{\begin{tikzpicture}[every node/.style=player,node distance=1cm]
	\skel
	\path (a1) edge (a2) (a2) edge (a3) (a3) edge (a4) (a4) edge (a1);
	\path (a1) edge (b1);
	\path (b2) edge (b3)
		(b3) edge (b4)
		(b4) edge (b5)
		(b5) edge (b6)
		(b6) edge (b7)
		(b7) edge (b1);
	\path (b7) edge (b2);
	\path (b3) edge (b8) (b8) edge (b6);
	\path (b2) edge (c1);
	\path (d1) edge (d2);
	\path (b6) edge (e1);
	\path (e1) edge (e2) (e2) edge (e3) (e3) edge (e1);
	\path (a4) edge (d1);
	\path (b4) edge (f1);
	\path (f1) edge (f2) 
		(f2) edge (f3)
		(f3) edge (f4);
\end{tikzpicture}}
\hfill
\subfloat[]%
[\label{fig:bridge-tree:2}The corresponding bridge tree $\tG$.]%
{\begin{tikzpicture}[node distance=1cm]
	\skel
	\begin{scope}[every node/.style=player]
	\node (B) [above right = 1cm and 1.5 cm of b1, label=above:{$7$}] {};
	\node (E) [below = of b6, label=right:{$3$}] {};
	\node (A) [below left = 1.3 cm of b1, label=right:{$4$}] {};
	\end{scope}
	\path (A) edge (d1) (d1) edge (d2);
	\path (A) edge (b1) (b1) edge (B);
	\path (B) edge (c1);
	\path (B) edge (E);
	\path (B) edge (f1)
		(f1) edge (f2)
		(f2) edge (f3)
		(f3) edge (f4);
	\begin{scope}[every node/.style=player]
	\node at (b1) { };
	\node at (c1) { };
	\node at (d1) { };
	\node at (d2) { };
	\node at (f1) { };
	\node at (f2) { };
	\node at (f3) { };
	\node at (f4) { };
	\end{scope}
\end{tikzpicture}}
\caption{%
	\label{fig:bridge-tree}%
	Bridge tree construction.
	Vertices representing components of more than $1$ vertices 
	have their number of vertices attached, here $4$, $7$, and $3$,
	respectively.}
\end{figure}%

\newpage
\section{Lower Bound on the Diameter}
\label{app:example-diameter}
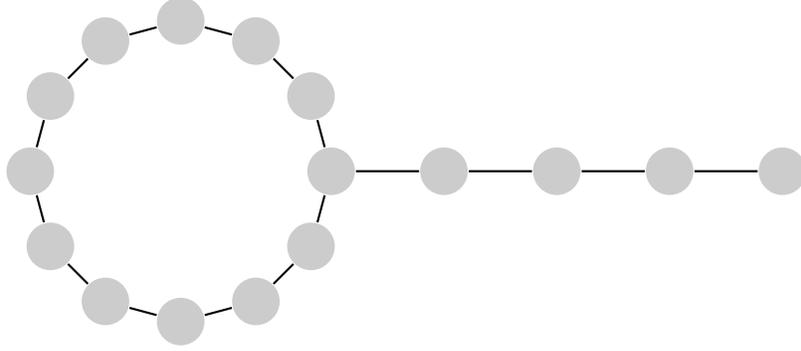
\begin{figure}
  \centering
  \begin{tikzpicture}[%
    every node/.style=player,%
    node distance=1.5cm,%
    ]
    \foreach \i in {0,30,...,330}{
    \node (v\i) at (\i:2cm) { };
    };
    \foreach \i in {0,30,...,300}{
    \pgfmathtruncatemacro{\next}{\i+30}
    \path (v\i) edge (v\next);
    };
    \path (v330) edge (v0);
    \node (r1) [right = of v0] { };
    \node (r2) [right = of r1] { };
    \node (r3) [right = of r2] { };
    \node (r4) [right = of r3] { };
    \path (v0) edge (r1) (r1) edge (r2) (r2) edge (r3) (r3) edge (r4);
  \end{tikzpicture}
  \caption{%
  \label{fig:bilateral-cycle-path}%
  Cycle with path attached, here $n=16$ and $\ell=4$.}
\end{figure}%
\begin{thmprop}
  For $\alpha=1$, the bridge tree can have diameter $\Omega(\sqrt{n})$,
  even if the graph is PS.
\end{thmprop}
\begin{proof}
  Consider a cycle with a path of length $\ell$ attached to it with one of its ends,
  as shown in \autoref{fig:bilateral-cycle-path}. 
  Diameter of the corresponding bridge tree is $\ell$.
  Let $n$ be the total number of vertices and
  let $\frac{1}{n} \frac{\ell \, (\ell+1)}{2}\leq 
  \alpha 
  \leq \frac{1}{n} \frac{((n-\ell)-1) \, (n-\ell)}{2}$;
  such an $\alpha$ exists if $n\geq 3\ell$.
  Because of the lower bound on $\alpha$,
  no vertex on the cycle wishes to connect to a vertex on the path,
  and also no vertex on the path wishes to connect to a vertex that is 
  located away from the cycle.
  Because of the upper bound on $\alpha$,
  it can also be shown easily that no two neighboring vertices 
  on the cycle wish to sell the link between them.
  Trivially, no link on the path will be sold.
  Hence this graph is PS.
  However, we can choose $n\geq 9$, $\ell\df\sfloor{\sqrt{n}}$, and $\alpha\df 1$,
  and so have a diameter of $\Omega(\sqrt{n})$ in the corresponding bridge tree.
\end{proof}
\begin{thmrem}
  In accordance with \autoref{thm:O1}, this example does not provide
  a price of anarchy beyond~$\Theta(1)$.
\end{thmrem}
\begin{proof}
  Social cost is 
  \begin{equation*}
    2n\al + \frac{1}{n}\,2\,\sum_{k=1}^\ell (n-k) \, k
    = 2n\al + \frac{\ell\,(\ell+1)}{n} \parens{n-\frac{2\ell+1}{3}}\enspace.
  \end{equation*}
  For a lower bound on the price of anarchy,
  we have to divide this by the optimum, which is $2n\al$ since $\al\leq n-1$.
  So we choose $\alpha$ as small as possible, \ie $\alpha=\frac{1}{n}\frac{\ell\,(\ell+1)}{2}$,
	hence $2n\al = \ell\,(\ell+1)$.
  It follows a lower bound on the price of anarchy of only:
  \begin{equation*}
    1 + \frac{1}{n} \parens{n-\frac{2\ell+1}{3}}
    = 1 + \parens{1-\frac{2\ell+1}{3n}} \leq 2 = \Theta(1)\period\qedhere
  \end{equation*}
\end{proof}

\end{document}